\documentclass[letterpaper, 10 pt, conference]{ieeeconf}  

\usepackage{graphicx}      
\usepackage{soul}
\usepackage{latexsym}
\usepackage{amsmath}
\usepackage{amsfonts}
\usepackage{epstopdf}
\usepackage{amssymb}
\usepackage[english]{babel}

\usepackage{graphicx,subfig}
\newtheorem{example}{Example}

\newtheorem{theorem}{Theorem}
\newtheorem{lemma}[theorem]{Lemma}
\newtheorem{proposition}[theorem]{Proposition}

\newtheorem{remark}[theorem]{Remark }

\usepackage{color}
\usepackage{algorithm}
\usepackage[noend]{algpseudocode}
 
\usepackage{xcolor}
\definecolor{verde}{rgb}{0,0.7,0}

\usepackage[colorlinks=true,linkcolor=blue,urlcolor=blue,citecolor=blue]{hyperref}

\newcommand{\be}{\begin{equation}}
\newcommand{\ee}{\end{equation}}

\newcommand{\bea}{\begin{eqnarray}}
\newcommand{\eea}{\end{eqnarray}}
\newcommand{\bean}{\begin{eqnarray*}}
\newcommand{\eean}{\end{eqnarray*}}

\IEEEoverridecommandlockouts                              

\begin{document}

\title{\LARGE \bf Modeling the Cooperative Process of Learning a Task}

 \author{Giulia De Pasquale and Maria Elena Valcher
 \thanks{Accepted for presentation at the European Control Conference (ECC 2022), London, UK. 
G. De Pasquale and M.E. Valcher are with
 the Dipartimento di Ingegneria dell'Informazione
 Universit\`a di Padova, 
    via Gradenigo 6B, 35131 Padova, Italy, e-mail:  \texttt{giulia.depasquale@phd.unipd.it, meme@dei.unipd.it}.}
   } 
 \maketitle

\begin{abstract}                          

 In this paper we propose a mathematical model for  a Transactive Memory System (TMS) involved in  the cooperative process of learning a task. The model is based on an intertwined dynamics involving both the individuals level of expertise and the interaction network among the cooperators. The model shows that if all the agents are non-stubborn, then all of them are able to  acquire the competence of the most expert members of the group, asymptotically reaching their level of proficiency. Conversely, when dealing with all stubborn agents, the capability to   pass on   the task depends on the connectedness properties of the interaction graph.
  \end{abstract}

\section{Introduction}
Nowadays, innovation and technological development strongly rely on the cooperation of individuals with different skills and from different majors. Consequently, 
mutual interactions among collaborative agents involved in the accomplishment of a single or of multiple tasks, together with their performance evaluation, have attracted the 
interest of researchers in many fields such as sociology \cite{ sociological_model,sociology1}, psychology \cite{Heider} and management \cite{sethi}. As a matter of fact, when dealing with complex and multidisciplinary duties that need to be efficiently accomplished,  it is important to consider both the 
actual skills of the 
team members and the interpersonal relationships among them \cite{TMS}.  In this regard, different \textit{group mind }theories have been formulated along the years, such as the Transactive Memory System (TMS) theory. It inspired this work and it deeply draws on the analogy between the mental operations of the individual and the process of the group  \cite{TMS}.

The literature on mathematical models aimed at formalizing the team working dynamics  at a macro-level, as a result of the micro-mechanisms occurring among the team members, is rather sparse. Oftentimes, the evolution of labor division and  of the mutual appraisal of  team member skills is treated in the models as   black box \cite{collectivelearning_TAC}. Consequently, the mathematical formalization of how the interpersonal relationships among the team members and their expertise evaluations affect the task accomplishment is still at an early stage  \cite{assign_and_appraise}. 

In this work we propose a mathematical model for the evolution of mutual appraisal and individual expertise of agents belonging to a collaborative team in charge of completing a specific task. We will focus on how the 
opinions that the members have about the abilities of their team mates
evolve along time, starting from an initial, possibly prejudicial, evaluation and updating it based on their effective 
skill levels in the task accomplishment.

As concerns     opinion dynamics  in a group of individuals, the amount of literature on which one can rely upon is quite broad. Most of the works in this context are related to topological characterizations of social networks  \cite{Altafini2013,DeGroot,Heider}.
 However, to the best of our knowledge, the influence of actual skills on the workload distribution within a team has only been addressed in 
  \cite{BulloML}, and no mathematical model describing how  the agents skills evolve, as a result of the perceived  performance of their team mates, has been proposed.

In Huang et al. \cite{assign_and_appraise}   a novel dynamical model for collaborative agents devising a task, with the objective of maximizing their performance, is proposed. The model describes a decentralized process by which each team member dynamically learns the abilities of its mates and consequently agrees with them on how to reorganize the task accomplishment. In Mei et al. \cite{collectivelearning_TAC} a dynamical model for the learning process by which cooperative agents carry on multiple tasks, sequentially executed,  is proposed. Different levels of complexity, starting from a centralized manager-based assignment process up to a completely distributed organization, are discussed. Finally, in Askarisichani et al. \cite{BulloML} the dynamics of small groups of individuals executing a sequence of collective tasks is studied. They observe that individuals with higher expertise are accorded  higher interpersonal influence, while low-performing agents are prone to underestimate their cooperators performance.
The objective of this paper   is to    replicate the dynamics of a group of members that  play out the TMS philosophy, according to which ``gaining entry to the group's stored knowledge is likely to be an efficient enterprise, even when one begins with a fairly inexpert member" \cite{TMS}. To do so we  study how the appraisal matrix of the network evolves over time as a convex combination of prejudices and effective expertise of the group members, and at the same time how the agents of the cooperative network improve in the task accomplishment, as a result of their interactions with more expert agents.  This is done in accordance with the theory stated in \cite{TMS}, in which it is assumed that each individual learns from those with a higher  level of expertise, even when the process  starts biased by some level of prejudice/reputation. In a TMS system, in fact, individual   performance  is related both to   their \textit{personal expertise} and to the \textit{circumstantial knowledge responsibility} that each member has towards the group. Our mathematical model shows that, under suitable   assumptions, each individual is able, through interactions with the other members,  to take advantage of the knowledge in the group and to gain from the others' levels of expertise in a way such that all members reach the same level of proficiency of the most expert individual in the group.

{\bf Notation}.\ ${\mathbb R}_+$ denotes the set of nonnegative real numbers.
 We let   ${\bf e}_i$ denote the $i$-th vector of the canonical basis of $\mathbb{R}^N$.
  The vectors ${\bf 1}_N$ and ${\bf 0}_N$ denote the $N$-dimensional vectors whose entries are all $1$ or $0$, respectively.
In the sequel, the $(i,j)$-th entry of a matrix  $M$ is denoted    by ${m}_{ij}$, while the $i$-th entry of a vector ${\bf v}$   by $v_i$.
Given two vectors ${\bf v}$ and ${\bf w}$ of the same size $N$, the expression
$\max \{{\bf v},{\bf w}\}$ denotes the $N$-dimensional vector ${\bf z}$ with
$z_i = \max \{v_i,w_i\}, i\in \{1,2,\dots, N\}.$
A matrix  $M$, in particular, a vector ${\bf v}$, is {\em nonnegative} ({\em positive}) if all its entries are nonnegative (positive). When so
 we adopt the notation $M \ge  0$ and ${\bf v}\ge 0$ ($M \gg 0$ and ${\bf v}\gg 0$).
A matrix $M\in {\mathbb R}_+^{N\times N}$, $N \ge 2$, is {\em irreducible} if there exists no permutation matrix $P\in {\mathbb R}_+^{N\times N}$ such that  $P^\top M P$ is block triangular,
otherwise it is called {\em reducible}.
Every reducible square matrix $M\in {\mathbb R}_+^{N\times N}$ can be brought, by means of a permutation matrix $P$, to {\em Frobenius form}, i.e.,
\be
P^\top M P = \begin{bmatrix}
M_{11} & 0 & \dots & 0 \cr
M_{21} & M_{22}& \dots & 0\cr
\vdots  & \vdots& \ddots & \vdots\cr
M_{k1} & M_{k2} & \dots & M_{kk}\end{bmatrix},
\label{Frobenius}
\ee
where each diagonal block $M_{ii}$ is either scalar or an irreducible matrix.
 A   matrix $M\in {\mathbb R}_+^{N\times N}$  is {\em row stochastic} if  
is a  nonnegative matrix
and $M{\bf 1}_N = {\bf 1}_N$. 
Given a matrix $M\in {\mathbb R}^{N\times N}$, the {\em spectrum} of $M$, $\sigma(M)$, is the set of eigenvalues of $M$.
We define the {\em spectral radius} of $M$ as $\rho(M) := \max\{|\lambda|: \lambda \in \sigma(M)\}$.\\
Given a vector ${\bf v}\in {\mathbb R}^N$, we define $\|{\bf v}\|_1 := \sum_{i=1}^N |v_i|$ and $\|{\bf v}\|_2 := \sqrt{{\bf v}^\top {\bf v}}$. 
 A {\em directed  graph} is a triple~\cite{Mohar} $\mathcal{G}=(\mathcal{V},\mathcal{E},{\mathcal A})$, where $\mathcal{V}= \{1,2,\dots, N\}$ is the set of vertices
 (nodes), $\mathcal{E}\subseteq\mathcal{V}\times\mathcal{V}$  the set of arcs (edges), and
${\mathcal A}=[a_{ij}]\in {\mathbb R}_+^{N\times N}$ the  {\em adjacency matrix} of the  graph $\mathcal{G}$. 
An arc $(j,i)$  belongs to ${\mathcal E}$ if and only if $a_{ij}
\ne 0$ and when so it has weight {\color{black}$a_{ij}>0$}.
 A sequence 
 ${j_1}
\rightarrow {j_2} \rightarrow  {j_3}  \rightarrow \dots  \rightarrow {j_{k}} \rightarrow {j_{k+1}}$
is a {\em path}  
 of length $k$ 
  from ${j_1}$ to ${j_{k+1}}$
provided that
$({j_1},{j_2}), ({j_2},{j_3}),\dots,$ $({j_{k}}, {j_{k+1}}) \in {\mathcal E}$. 

We will oftentimes use the notation   ${\mathcal D}({\mathcal A})$  to denote the directed graph having ${\mathcal A}\in {\mathbb R}_+^{N\times N}$ as adjacency matrix. 
A directed graph is said to be {\em strongly connected} if for every pair of vertices $i$ and $j$ there is a path from $i$ to $j$ and from $j$ to $i$.
${\mathcal D}({\mathcal A})$ is strongly connected if and only if ${\mathcal A}$ is irreducible.
If the adjacency  matrix ${\mathcal A}$ is in Frobenius form \eqref{Frobenius} with $k$ diagonal blocks ${\mathcal A}_{ii}$ that are either scalar or irreducible matrices, then we can partition the set of vertices ${\mathcal V}=\{1,2,\dots,N\}$ into $k$ 
{\em communication classes} ${\mathcal C}_i, i\in \{1,2,\dots,k\},$ where ${\mathcal C}_i$ is the set of nodes $\{(\sum_{h=0}^{i-1} n_h) +1, (\sum_{h=0}^{i-1} n_h) +2, \dots, (\sum_{h=0}^{i} n_h) \}$ (with $n_0 :=0$ and $n_i := |{\mathcal C}_i|$ for $i\in \{1,2,\dots, k\}$) corresponding to the  row/column indices of the  (entries of the) diagonal block ${\mathcal A}_{ii}$.
For $j <i$ the class ${\mathcal C}_i$ is {\em accessible} from the class ${\mathcal C}_j$  (for short, ${\mathcal C}_j \to {\mathcal C}_i$) if there is a path from some node of 
${\mathcal C}_j$  to some node of ${\mathcal C}_i$. Clearly, ${\mathcal C}_i$ is  accessible from  itself, while for $j >i$ the class ${\mathcal C}_i$ is never accessible from   ${\mathcal C}_j$.

\section{The  Cooperative Learning Model}

We consider a team consisting of $N$ agents who need to collectively perform a task, {\color{black}e.g., paving a floor, laying some tiles, etc. The team needs to perform this task daily,  but team members  exhibit   different levels of expertise and hence different level of involvement in the overall task execution.}
We denote by $y_i(t)$ the expertise level of the agent $i$ at the time {\color{black} (day)} $t$, $i \in \{1,2,\dots, N\}$, compared 
to the  expertise level required to perfectly perform the job. This amounts to saying that at each time $t$, $y_i(t)\in [0,1]$.
For every $i,j \in \{1,\dots, N\}$, we let $m_{ij}(t)$ denote the 
percentage of the task that agent $i$ believes should be attributed to agent $j$. This can be viewed as the
level of competence that agent $i$ attributes to   agent $j$  at the time  $t$ compared to the others. 
These two entities dynamically evolve, in an intertwined way, as follows
\begin{align}
\!\!\!m_{ij}(t+1) &= (1-\lambda_i) m_{ij}(t) + \lambda_i \frac{y_j(t)}{\sum_{i=1}^N y_i(t)}\label{mij}\\
\!\!\!y_i(t+1) &= y_i(t) + \ell_i \cdot \max  \big\{\sum_{k=1}^N m_{ik}(t)y_k(t)-y_i(t), 0\big\}  \label{yi}
\end{align}
where $\lambda_i \in [0,1]$ is a  coefficient that weights how much the  opinion of the agent $i$ about the competence level of the other agents is anchored to its initial evaluation (prejudice). 
In the general case, $m_{ij}$ at time $t+1$  is the convex combination of its   value at time $t$ and of the   expertise level of agent $j$ at time $t$ compared with the total expertise of the team members.
For $\lambda_i = 0$, agent $i$ totally ignores the relative level of experience of the other agents within the overall team and fully relies on its prejudice, while for $\lambda_i = 1$ agent $i$ has an objective evaluation of the expertise level of each team member, and is devoid of prejudices. 
As far as equation \eqref{yi}  is concerned,  note that
 if $y_i$ is higher  than its perceived average expertise level of the team, no action is taken. If it is lower,  the agent tries to increase its expertise level 
to such  a perceived average level, compatibly with its learning capabilities. Indeed, 
the parameter $\ell_i$ represents the \textit{learning coefficient} of agent $i$, and it belongs to the interval $(0,1]$.

Equations \eqref{mij}-\eqref{yi} can be expressed in matrix form as
\begin{align}
M(t+1) &= (I_N-\Lambda)M(t) + \Lambda \frac{\mathbf{1}_N {\bf y}^\top(t)}{\lVert {\bf y}(t) \lVert_1} \label{m},\\
{\bf y}(t+1) &= {\bf y}(t) + L \cdot \max  \{ M(t){\bf y}(t)-{\bf y}(t), 0\}, \label{y}
\end{align}
 where $M(t) \in [0,1]^{N \times N}$ is the matrix whose $(i,j)$-th entry is $m_{ij}(t)$ and ${\bf y}(t)\in [0,1]^{N}$ is the vector  {\color{black} whose $i$-th entry is} $y_i(t)$. The $\max $ function, in equation \eqref{y}, acts component-wise, as clarified in the Notation. 
The diagonal matrices $L =[\ell_i]$ and $\Lambda = [\lambda_i]$, $i \in \{1,\dots, N\}$, are the $N \times N$ diagonal matrices whose $i$-th diagonal entries correspond to the learning coefficient and the perception coefficient, respectively, of the $i$-th agent.  \\  
   Equation \eqref{m} is very similar to the ``differentiation model" (D model) proposed in \cite{BulloML}. 
  In  \cite{BulloML}  the influence matrix $M$ at  time $t+1$ is assumed to be the result of a scalar convex combination of the matrix itself at   time   $t$ and of a function of the vector ${\bf y}(t)$ (in other words, $\lambda_i=\lambda$ for every agent $i$, and the parameter $\lambda$  has a different interpretation, since it is employed in the time-scale   adjustment of the dynamics). On the other hand,  \eqref{y} proposed  in this work is a dynamical model that describes the time evolution of the expertise vector, while in \cite{BulloML} expertise is   estimated from data.

 This work is developed under the following:

\textbf{Assumption 1} [\textit{Initial Conditions}].  Each agent's initial expertise takes values in $(0,1]$, namely ${\bf y}(0) \in (0,1]^{N}$, and the matrix $M(0)\in \mathbb{R}^{N \times N}$ is   row stochastic.
\smallskip

The   assumption that the vector ${\bf y}(t)$ is strictly positive excludes the (unrealistic) presence of totally inexpert individuals in the  {\color{black} team} (corresponding to $y_i(0) =0$).
On the other hand, the interpretation of the entries of the matrix $M(t)$ naturally corresponds to the fact that $M(t)$ is row stochastic at every time $t$, and hence in particular at $t=0$.

The following proposition shows that, under Assumption 1,  the proposed model is well posed.

\begin{proposition} 

If Assumption 1 holds, then
\begin{itemize}
\item[i)]   $M(t) \in   \mathbb{R}^{N \times N}$ is row stochastic, $\forall t \ge 0$,
\item[ii)]  ${\bf y}(t)\in (0,1]^N$, $\forall t \geq 0$.
\end{itemize}
\end{proposition}

\begin{proof} {\color{black}We first observe that equation \eqref{y} guarantees that, under Assumption 1, ${\bf y}(t)\gg 0, \forall t\ge 0$}. \\ 
i) We proceed by induction on $t\ge 0$.
$M(0)$ is row stochastic by assumption. 
We now show that  if $M(t)$ is row stochastic, then $M(t+1)$ is row stochastic as well.  
{\color{black}First of all, for every index $i$, ${\bf e}_i^\top M(t)$ and ${\bf y}^\top (t)$ are nonnegative vectors and neither of them can be zero. Therefore ${\bf e}_i^\top M(t+1)$ is in turn a nonnegative nonzero vector. Also,}
by  equation \eqref{m}, we have 
\begin{align*}
M(t+1)\mathbf{1}_N  &= (I_N-\Lambda)M(t)\mathbf{1}_N  + \Lambda \frac{\mathbf{1}_N y^\top(t)}{ \lVert {\bf y}(t) \lVert_1}\mathbf{1}_N \\
&= (I_N-\Lambda)\mathbf{1}_N +\Lambda \mathbf{1}_N = \mathbf{1}_N.
\end{align*}

\noindent ii)\
Also in this case, we proceed by induction on $t\ge 0$, by making use of Assumption 1 and of the previous part i).
By Assumption 1, the result is  true for $t=0$. Now, let us assume that ${\bf y}(t)\in (0,1]^N$ and $M(t)$ is row stochastic, and show that ${\bf y}(t+1)\in (0,1]^N$.
From  equation \eqref{yi} one gets that if $i\in \{1,\dots,N\}$ is such that $\sum_{k = 1}^N m_{ik}(t)y_k(t)-y_i(t) \le 0$, then  $y_i(t+1) = y_i(t)\in (0,1]$, while
if  $i\in [1,N]$ is such that $\sum_{k = 1}^N m_{ik}(t)y_k(t)-y_i(t) > 0$, one has
 \begin{align}
y_i(t+1) 
 (1-\ell_i)y_i(t) + \ell_i \sum_{k = 1}^N m_{ik}(t)y_k(t). \label{wellpo}
\end{align}
Since $\ell_i\in(0,1]$ and $\sum_{k = 1}^N m_{ik}(t)=1$, the last expression in \eqref{wellpo} is a convex combination of   two terms, $y_i(t)$ and $\sum_{k = 1}^N m_{ik}(t)y_k(t)$, both of them belonging to   $(0,1]$, by the inductive assumption, and hence it belongs  to $(0,1]$, too. 
$\square$
\end{proof}

 Lemma \ref{monotone}, below, gives some insight on how the components of the vector ${\bf y}(t)$ in \eqref{y} evolve over time. In particular, it shows that all its  entries exhibit   a non-decreasing trend, and their maximum value at $t=0$ remains stationary.

\begin{lemma}\label{monotone}

Under Assumption 1, for every $t\ge 0$ we have 
\begin{itemize}
\item[i)] $y_i(t+1) \geq y_i(t)\geq y_i(0)$, for every $i\in \{1,2,\dots, N\}$;
\item[ii)] $\max _k y_k(t+1) = \max _ky_k(t) = \max _ky_k(0)$.
\end{itemize}
\end{lemma}

\begin{proof}  From equation \eqref{yi} it clearly follows that, $\forall i \in \{1, \dots, N\}$, 
   if  $y_i(t) \geq \sum_{k=1}^Nm_{ik}(t)y_k(t)$ then $y_i(t+1) = y_i(t)$, otherwise 
$y_i(t+1) > y_i(t).$ This proves i).
 
On the other hand in this latter case, namely when  $y_i(t) < \sum_{k=1}^Nm_{ik}(t)y_k(t)$, then
\begin{align*}
&y_i(t+1) = y_i(t)+ \ell_i \Big(\sum_{k=1}^Nm_{ik}(t)y_k(t)-y_i(t) \Big) \\
&{\color{black}\leq} 
 (1-\ell_i)y_i(t)+\ell_i(\max _k y_k(t))\leq\max _k y_k(t).
\end{align*}
Therefore,  in particular, $\max _k y_k(t+1) \leq \max _k y_k(t)$. 
Finally,  if the agent $i$ is such that $y_i(t)=\max _k y_k(t)$ then, clearly $y_i(t) \ge \sum_{k=1}^Nm_{ik}(t)y_k(t)$ and hence   $y_i(t+1) = y_i(t) = \max _k y_k(t)$.
Therefore 
statement ii) holds true.
$\square$
\end{proof}

The main contribution   of Lemma \ref{monotone} is to prove that, according to model \eqref{yi}, every agent over time can only increase its expertise or, in the worst case, leave it unaltered. On the other hand, model \eqref{yi} also formalizes the fact 
  that in a closed team no agent  can reach   an expertise level higher than the highest expertise level that one of its members exhibited at the initial time.

\section{Equilibria and asymptotic behavior when none of  the agents is   stubborn}

Model \eqref{m}-\eqref{y} can evolve in quite different ways and  asymptotically reach very diverse configurations, depending 
on the mutual attitude of the team workers. The crucial parameters are the coefficients $\lambda_i$'s that express how much agent $i$ is   open minded
and  updates its evaluations of the team mates based on their actual skills.
Three possible scenarios can be considered: the case when none of the agents is stubborn 
($\lambda_i\ne 0$ for every $i$), the case when they are all stubborn ($\Lambda =0$), and the intermediate case when only a subgroup of the agents is stubborn.
Due to space constraints, in the paper we will address only the first two opposite situations.\\
Specifically, in this section we investigate the structure of the equilibrium points of the model \eqref{m}-\eqref{y} 
under the following:

\textbf{Assumption 2} [{\it No agent is  stubborn}]\ The matrix $\Lambda$   appearing in equation \eqref{m} has diagonal entries   $\lambda_i \in (0,1]$, $\forall i \in \{1,\dots,N\}$.

 \begin{proposition}\label{equilibria}
Under Assumptions 1 and 2, a pair $(\bar M, \bar {\bf y})$, with $\bar M\in {\mathbb R}^{N\times N}$ and $\bar {\bf y}\in (0,1]^N$,  is an equilibrium point of the model \eqref{m}-\eqref{y}
if and only if 
\be 
\bar M= \frac{1}{N} {\bf 1}_N {\bf 1}_N^\top,
\qquad {\rm and}\qquad
\bar {\bf y} = \alpha {\bf 1}_N,
\label{equ_bar1}
\ee
for some $\alpha\in (0,1]$.
\end{proposition}

\begin{proof}
A pair $(\bar M, \bar {\bf y})$ is an equilibrium point of the model \eqref{m}-\eqref{y}
if and only if
\begin{align*}
\bar M &= (I_N-\Lambda)\bar M + \Lambda \frac{\mathbf{1}_N \bar {\bf y}^\top}{\lVert \bar {\bf y} \lVert_1},\\
\bar {\bf y} &= \bar {\bf y} + L \cdot \max  \{ \bar M \bar {\bf y}-\bar {\bf y}, 0\},
\end{align*}
which is equivalent (by the nonsingularity of   $\Lambda$ and $L$) to
\begin{align}
&\bar M = \frac{\mathbf{1}_N \bar {\bf y}^\top}{\lVert \bar {\bf y} \lVert_1} \label{m1},\\
& \bar M \bar {\bf y}-\bar {\bf y} \le 0. \label{y1}
\end{align}
It is immediate to see that if $\bar {\bf y}$ and $\bar M$ take the expression in \eqref{equ_bar1} then the previous identities hold.\\
Conversely,   suppose that the pair $(\bar M,\bar {\bf y})$, with $\bar M\in {\mathbb R}^{N\times N}$ and $\bar {\bf y}\in (0,1]^N$,  is an equilibrium point and hence satisfies
\eqref{m1}-\eqref{y1}.
Replace \eqref{m1} in \eqref{y1} thus getting
$$
\frac{\mathbf{1}_N \|\bar {\bf y}\|_2^2}{\lVert \bar {\bf y} \lVert_1}   = \frac{\mathbf{1}_N \bar {\bf y}^\top}{\lVert \bar {\bf y} \lVert_1} \bar {\bf y} \le \bar {\bf y},
$$
or, componentwise,
\begin{equation}
 y_i \geq \frac{\sum_{i=1}^N y_i^2}{\sum_{i=1}^N y_i}, \qquad \forall i\in \{1,\dots, N\}.
 \label{per_comp1}
\end{equation}
We will now show that if the previous inequalities hold, then $\bar {\bf y} = \alpha {\bf 1}_N$, $\exists \alpha \in (0,1]$.
To do so, we proceed by contradiction. Let us assume that the components of the vector $\bar {\bf y}$ are not all   identical. Without loss of generality, let us assume $y_1 \leq y_j$, $\forall j \neq 1$, and that there exists $k \in \{2, \dots, N\}$ such that $y_1 < y_k$. We have that\eqref{per_comp1} implies
$y_1 (y_2+ \dots + y_N) \geq y_2^2 + \dots +y_N^2 $, but from the fact that $y_1y_j \leq y_j^2$, $\forall j \neq 1$ and that $y_1y_k < y_k^2$ we get 
$$y_1(y_2+\dots + y_N)<y_2^2+ \dots + y_k^2+ \dots + y_N^2,$$ 
thus leading to a contradiction. 
Therefore all entries of $\bar {\bf y}$ must be identical.
So, if $\bar {\bf y} = \alpha {\bf 1}_N$, then
$$ \bar M = \frac{\mathbf{1}_N \bar {\bf y}^\top}{\lVert \bar {\bf y} \lVert_1} = \frac{1}{N} {\bf 1}_N {\bf 1}_N^\top,$$
and this completes the proof.
$\square$
\end{proof}

\begin{remark}
Note that the   $M(0)$ is not necessarily irreducible. However, as a consequence of Assumption 1 and Assumption 2, 
$M(t)$ becomes irreducible {\color{black}from $t=1$ onward.}
\end{remark}

We now show that  
 in this scenario, namely under Assumption 2,
for every choice of the initial conditions satisfying Assumption 1,  
all the expertise levels of the agents in the team asymptotically converge to the same value.

 \begin{proposition}\label{y_convergence}
Under Assumptions 1 and 2, the vector sequence ${\bf y}(t),  t\in {\mathbb Z}_+,$  asymptotically converges to 
$\alpha {\bf 1}_N$, where $\alpha = \max _i y_i(0) \in \mathbb{R}_+$. 

\end{proposition}  

\begin{proof} From Lemma \ref{monotone} we know that the vector sequence $\{ {\bf y}(t)\}_{t\in {\mathbb Z}_+}$ is monotone and upper bounded by the vector 
$\alpha {\bf 1}_N$, where $\alpha = \max _i y_i(0)$, and hence it converges. We want to prove that $\bar {\bf y}:= \lim_{t\to +\infty} {\bf y}(t)$ coincides with $\alpha {\bf 1}_N$.
If ${\bf y}(0)= \alpha {\bf 1}_N$ the result is obvious, since $\alpha {\bf 1}_N$ is an equilibrium point of \eqref{y} for every choice of the row stochastic matrix $M(t)$. 
Suppose, on the contrary, that ${\bf y}(0) \le \bar {\bf y} \le  \alpha {\bf 1}_N$, but  $\bar {\bf y} \ne \alpha {\bf 1}_N$.
This means that there exists $i\in \{1,2,\dots, N\}$ such that $\bar {\bf y}_i< \alpha$. It entails no loss of generality assuming that $i = {\rm argmin}_k \bar {\bf y}_k$. On the other hand, we also know that there exists $j\in \{1,2,\dots, N\}$ such that $\bar {\bf y}_j = \alpha$.\\
Now we observe that  condition ${\bf y}(0) \in (0,1]^N$ from Assumption 1, together with Assumption 2, ensure that $M(t) \gg 0$ for every $t\ge 1$, and hence, in particular, ${\bf e}_i^\top M(t) \gg 0$ for every $t\ge 1$. Therefore 
$\max  \big\{\sum_{k=1}^N m_{ik}(t)\bar {\bf y}_k-\bar {\bf y}_i, 0\big\} = \sum_{k=1}^N m_{ik}(t) \bar {\bf y}_k-\bar {\bf y}_i > 0$ and this shows that if $\bar {\bf y} \ne \alpha {\bf 1}_N$, it cannot be 
the limit of the sequence ${\bf y}(t), t\ge 0$.
$\square$
\end{proof}

In the following we discuss the asymptotic behavior of the matrix $M(t)$ involved in the model dynamics \eqref{m}-\eqref{y}.

\begin{proposition} \label{M_convergence}
Under Assumptions 1 and 2, 
$$M_{\infty} := \lim_{t\rightarrow \infty}M(t) = \frac{{\bf 1}_N{\bf 1}_N^\top}{N}.$$
\end{proposition}
\begin{proof}
 If we define $u_j(t):= \frac{y_j(t)}{\lVert {\bf y}(t)\lVert_1}$, equation \eqref{mij} takes the form
$m_{ij}(t+1) = (1-\lambda_i)m_{ij}(t)+\lambda_iu_j(t)$
from which it follows that 
$$m_{ij}(t) = (1-\lambda_i)^{t}m_{ij}(0)+ \lambda_i \sum_{k=0}^{t-1}(1-\lambda_i)^{t-1-k}u_{j}(k)$$
and 
\begin{align*}
\lim_{t\rightarrow \infty} m_{ij}(t) &= \lim_{t \rightarrow \infty} \Big[ (1-\lambda_i)^t m_{ij}(0) + \lambda_i \sum_{k=0}^{t-1} (1-\lambda_i)^k u_j(k)\Big] \\
&= \lambda_i \sum_{k=0}^{\infty}(1-\lambda_i)^k u_j(k)
\end{align*}
where the last inequality follows from the fact that $\lim_{t \rightarrow \infty}(1-\lambda_i)^t m_{ij}(0) =0$ for $0<\lambda_i \leq 1$. So, the proof of existence of $\lim_{t\rightarrow \infty} m_{ij}(t)$   follows from the proof of convergence of the time series $\sum_{k=0}^{\infty}(1-\lambda_i)^ku_j(k)$ 
which is shown in the Appendix (see Lemma \ref{innominato}  in the Appendix).
 So, we conclude that when $\Lambda$  has no zero diagonal entries and hence is nonsingular,  the matrix $M(t)$ asymptotically converges to some matrix 
 $M_{\infty}$.
 Clearly, from \eqref{m} one gets
$$ \lim_{t \rightarrow \infty} M(t+1) = \lim_{t \rightarrow \infty}  (I_N-\Lambda)M(t)+\Lambda \frac{{\bf 1}_N {\bf y}^\top(t)}{{\bf 1}_N^\top {\bf y}(t)}, $$
and hence,   by Lemma \ref{monotone} and Proposition \ref{y_convergence}, it must be true that
$M_{\infty} = (I_N-\Lambda M_\infty)+\Lambda \frac{{\bf 1}_N {\bf 1}_N \alpha}{N \alpha}$.
Finally, Assumption 2  leads to $M_\infty = \frac{{\bf 1}_N{\bf 1}_N^\top}{N}$.
$\square$
\end{proof}

{\color{black} Propositions \ref{y_convergence} and \ref{M_convergence} together lead to}  the following result.

\begin{theorem} 
Under Assumptions 1 and 2, for every pair $(M(0), {\bf y}(0)) \in \mathbb{R}^{N \times N} \times \mathbb{R}$, 
 the sequence $(M(t), {\bf y}(t)), t\in {\mathbb Z}_+,$ generated by the   model \eqref{m}-\eqref{y} converges to the pair $(\bar M,\bar {\bf y})$, with 
$\bar M =  \frac{{\bf 1}_N {\bf 1}_N^\top}{N} \text{ and } \bar {\bf y} = \alpha {\bf 1}_N, $
where $\alpha = \max_i y_i(0)$.
\end{theorem}

\begin{example}
Consider  a team consisting of $N = 15$ agents and select ${\bf y}(0)$   from a uniform distribution in the interval $(0,1]$, ${\bf y}(0) \sim {\mathcal U}(0,1]^N$ and $M(0)$ as  a row stochastic matrix, with $M(0) \sim \mathcal{U}[0,1]^{N\times N}$. The sequence  $(M(t), {\bf y}(t)), t\in {\mathbb Z}_+,$ generated by the   model \eqref{m}-\eqref{y} 
starting from the randomly selected initial pair $(M(0), {\bf y}(0))$,  converges, after approximately $120$, units of time to the equilibrium point   $(\bar{M}, \bar{\bf y})$ with $\bar{\bf y} = \alpha {\bf 1}_N,$  $\alpha = {\rm max}_i y_i(0) = 0.969$, and $\bar{M} = \frac{{\bf 1}_N{\bf 1}_N^\top}{15}$.
\begin{figure}[h]
 \centering
 \includegraphics[scale = 0.20]{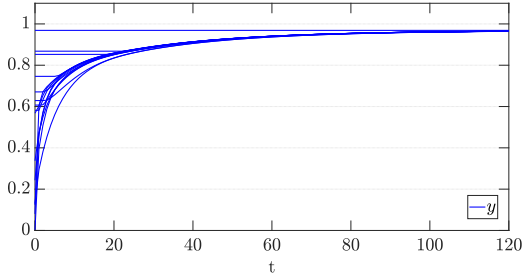}
\caption{Simulation of the dynamics of ${\bf y}(t)$ in the Cooperative Learning  model with no stubborn agents.}
\end{figure}
\end{example}

\section{Equilibria and asymptotic behavior when when all  agents are  stubborn}
In this section we focus on the study of the dynamical model \eqref{m}-\eqref{y}  assuming that  all agents are  stubborn, namely  
$\Lambda =0$. 
When this is the case the system model becomes
\begin{align}
M(t) &= M(0) \label{m0},\\
{\bf y}(t+1) &= {\bf y}(t) + L \cdot \max  \{ M(0){\bf y}(t)-{\bf y}(t), 0\}, \label{y0}
\end{align}
 where $M(0) \in [0,1]^{N \times N}$ is row stochastic and ${\bf y}(t)\in (0,1]^{N}$ for every $t\in {\mathbb Z}_+$.
Since the matrix $M(0)$ is constant, the model evolution reduces to the evolution of the vector ${\bf y}(t), t\in {\mathbb Z}_+$, representing the skills of the team members.
The structure of the directed graph associated with $M(0)$ has a strong impact on the asymptotic evolution of  the vector sequence ${\bf y}(t), t\in {\mathbb Z}_+$
that, however, is always monotonically increasing and upper bounded by $\alpha {\bf 1}_N$, where $\alpha := \max_i y_i(0)$
 (see Lemma \ref{monotone}).

 We first consider the case when $M(0)$ is irreducible, namely  ${\mathcal D}(M(0))$ is strongly connected.

 \begin{proposition}\label{M0irreducible}
 Consider the model \eqref{m0}-\eqref{y0} under Assumption 1. If 
  $M(0)$ is irreducible, then
 $\bar {\bf y} := \lim_{t\to \infty} {\bf y}(t) = \alpha {\bf 1}_N,$
 where $\alpha = \max_i y_i(0)$.
\end{proposition}

 \begin{proof} 
 The existence of the limit, $\bar {\bf y}$, follows by the same reasoning adopted in the proof of Proposition \ref{y_convergence}. It is immediate to see that $\bar {\bf y}$ is an equilibrium point of \eqref{y0},  and satisfies $M(0)\bar {\bf y} \le \bar {\bf y}$.
Since $M(0)$ is irreducible and row stochastic, by Lemma \ref{MYirr}  in the Appendix,  all entries of $\bar {\bf y}$ must be the same and $\bar {\bf y}$ must be an eigenvector of $M(0)$ corresponding to $1$.
So, by Lemma \ref{monotone}, part ii), 
$\bar {\bf y}$ must coincide with $\alpha {\bf 1}_N,$
 where $\alpha = \max_i y_i(0)$.
 $\square$
 \end{proof}

We now address the case  when $M(0)$ is not irreducible, namely the directed graph representing the interactions and mutual evaluations within the team of agents is not strongly connected. If so, $M(0)$ can always be reduced to   Frobenius form by means of a simple relabelling of the agents.

 \begin{proposition}
 \label{M0kblocchi}
 Consider the model \eqref{m0}-\eqref{y0} under Assumption 1 and assume that (possibly after a relabelling)
$M(0)$ takes the following structure: 
\be
M(0)=\begin{bmatrix}
M_{11} & 0 & \dots & 0 \cr
M_{21} & M_{22}& \dots & 0\cr
\vdots  & \vdots& \ddots & \vdots\cr
M_{k1} & M_{k2} & \dots & M_{kk}\end{bmatrix},
\label{kblocchi}
\ee
where each diagonal block $M_{ii}, i\in \{1,2,\dots,k\},$ is either scalar or an irreducible matrix of size $n_i\times n_i$ ($\sum_{i=1}^k n_i=N$).
Accordingly block-partition the vector ${\bf y}(t)$ as
${\bf y}(t)= [{\bf y}_1(t) {\bf y}_2(t)  \dots  {\bf y}_k(t)]^\top,$
with ${\bf y}_i(t)\in (0,1]^{n_i}$,
and set
$\alpha_i := \max_\ell [{\bf y}_i(0)]_\ell, \, i\in\{1, \dots, k\}.$
Let ${\mathcal C}_i, i\in \{1,2,\dots, k\}$, be the communication class associated with the diagonal block $M_{ii}$.
Then, for every $i\in \{1,2,\dots, k\}$, 
 \be
\left( \min_{j\le i: {\mathcal C}_j \to {\mathcal C}_i} \alpha_j\right) {\bf 1}_{n_i} \le  \lim_{t\to \infty} {\bf y}_i(t) 
\le \left( \max_{j\le i: {\mathcal C}_j \to {\mathcal C}_i} \alpha_j\right) {\bf 1}_{n_i}.
\label{boundskblocchi}
\ee
\end{proposition}

\begin{proof}
We first observe that for every $i\in \{1,2,\dots, k\}$ the behavior of ${\bf y}_i(t)$ only depends on the vectors ${\bf y}_j(\tau)$, with $\tau\in \{0,1,\dots, t-1\}$ and
$j \le i$, such that the corresponding class ${\mathcal C}_j$ has access to class ${\mathcal C}_i$.
So, for every $i\in \{1,2,\dots, k\}$, in order to determine ${\bf y}_i(t)$ we can always permute the blocks of $M(0)$ or equivalently the classes ${\mathcal C}_j$  (and the blocks of ${\bf y}(t)$) and restrict our attention to a lower dimensional model. Specifically, set
${\mathcal A}(i) := \{j\le i: {\mathcal C}_j \to {\mathcal C}_i\} \quad {\rm and}\quad  s(i) := | {\mathcal A}(i) |,$
and consider the submatrix
$$\tilde M^{(i)}=\begin{bmatrix}
\tilde M_{11} & 0 & \dots & 0 \cr
\tilde M_{21} & \tilde M_{22}& \dots & 0\cr
\vdots  & \vdots& \ddots & \vdots\cr
\tilde M_{s(i)i1} & \tilde M_{s(i)2} & \dots & \tilde M_{s(i)s(i)}\end{bmatrix},$$
obtained by first permuting the blocks of $M$ so that the first $s(i)$ correspond to the classes in ${\mathcal A}(i)$ (this permutation results in 
a new Frobenius form) and then restricting the dynamics to the first $s(i)$ blocks.
Note that $ \tilde M_{s(i)s(i)} = M_{ii}$, and if we denote by   $\tilde{\mathcal C}_j$ the class associated with $\tilde M_{jj}$, then $\tilde{\mathcal C}_j$ has access to $\tilde {\mathcal C}_{s(i)}$, the last class, for every $j\le s(i)$.
\\
We accordingly denote by 
{\small
$$\begin{bmatrix} \tilde {\bf y}_1(t) \cr \vdots \cr \tilde {\bf y}_{s(i)-1}(t) \cr \tilde  {\bf y}_{s(i)}(t)\end{bmatrix}, 
\quad {\rm with} \quad \tilde  {\bf y}_{s(i)}(t)=    {\bf y}_{i}(t),$$}

the subvector of  ${\bf y}(t)$ corresponding to the blocks in $\tilde M^{(i)}$.
\\
We can now prove the result by induction on $s(i)$. If $s(i)=1$, namely the class ${\mathcal C}_i$ is the only class having access to itself, then
$\tilde M_{s(i)s(i)} =  M_{ii}$ is either a nonzero scalar (in fact, it is equal to $1$ by the row stochasticity assumption on $M(0)$)  or irreducible, and by Proposition \ref{M0irreducible} one can claim 
that \eqref{boundskblocchi} holds since $\alpha_i {\bf 1}_{n_i} =  \lim_{t\to \infty} {\bf y}_i(t)$, and hence
$\alpha_i {\bf 1}_{n_i} \le  \lim_{t\to \infty} {\bf y}_i(t) \le \alpha_i {\bf 1}_{n_i}.$
Suppose now that the result is true for $s(i)-1$. We want to prove that the result holds for $s(i)$.
If we set $\bar {\bf y}_j :=  \lim_{t\to \infty} \tilde {\bf y}_j(t)$, then
the limit vector  $[ \bar {\bf y}_1^\top\  \bar {\bf y}_2^\top\ \dots\ \bar {\bf y}_{s(i)}^\top]^\top$
satisfies
\be
\sum_{j=1}^{s(i)} \tilde M_{s(i)j} \bar {\bf y}_j \le \bar {\bf y}_{s(i)}.
\label{cocca}
\ee
By the inductive assumption, for every $j <s(i)$ we have
$$\left( \min_{h\le j: \tilde{\mathcal C}_h \to \tilde{\mathcal C}_j} \tilde \alpha_h\right) {\bf 1}_{n_j} \le  \bar {\bf y}_j
\le \left( \max_{h\le j: \tilde {\mathcal C}_h \to \tilde {\mathcal C}_j} \tilde \alpha_h\right) {\bf 1}_{n_j},$$
where 
$\tilde \alpha_h := \max_\ell [\tilde {\bf y}_h(0)]_\ell,$
and hence, a fortiori,
\be
\left( \min_{h\le s(i)} \tilde\alpha_h\right) {\bf 1}_{n_j} \le  \bar {\bf y}_{j}
\le \left( \max_{h\le s(i)} \tilde \alpha_h\right) {\bf 1}_{n_j}.
 \label{cocca2}
  \ee
By making use of the first inequality in \eqref{cocca2}, \eqref{cocca} leads to 
$$
  \left( \min_{h\le s(i)} \tilde \alpha_h\right) (I_{n_i} -   M_{ii}) {\bf 1}_{n_i} \le 
  (I_{n_i} -   M_{ii}) \bar {\bf y}_{s(i)}.
 $$
  By making use, again, of the fact that $ M_{ii}$ is an irreducible stochastic submatrix  and hence $(I_{n_i} -  M_{ii})^{-1}$ exists and is a strictly positive matrix, we deduce that
 $
  \left( \min_{h\le i} \tilde \alpha_h\right)   {\bf 1}_{n_i} \le 
 \bar {\bf y}_{s(i)}.$
  And this proves
 $\left( \min_{j\le i: {\mathcal C}_j \to {\mathcal C}_i} \alpha_j\right) {\bf 1}_{n_i} \le  \lim_{t\to \infty} {\bf y}_i(t).$
 On the other hand, Lemma \ref{monotone} allows to say that 
 $ \lim_{t\to \infty} {\bf y}_i(t) \le  \left( \max_{h\le i} \tilde \alpha_h\right)   {\bf 1}_{n_i} ,$
 and this proves the second inequality in \eqref{boundskblocchi}.
 $\square$
 \end{proof}

 \medskip

\section{Conclusions}
When there is no stubborn agent in the network all  team members  reach the same level of expertise in accomplishing the task as the most expert agent at the beginning of the execution. On the contrary, if all the agents are   stubborn 
 we are only able to provide lower and upper bounds on the asymptotic skill levels achieved by each single communication
 class in the network.
 The general case,   in which stubborn and non-stubborn agents cooperate, is much more involved and will be presented in a future work.
Future research will   try to overcome the limitation of our model that assumes that all individuals  have full information on the expertise levels of
all the other agents.  

\appendix

\section*{Technical Lemmas}

\begin{lemma} \label{innominato}  Let $\lambda\in (0,1]$ and let ${\bf y}(t), t\in {\mathbb Z}_+$, be the sequence generated by model \eqref{m}-\eqref{y}, under Assumptions 1 and 2.
For any $j\in \{1,2,\dots, N\}$, set
$u_j(t) := \frac{y_j(t)}{\lVert {\bf y}(t)\lVert_1}, \qquad t\in {\mathbb Z}_+.$
The time series 
$\sum_{t=0}^\infty (1-\lambda)^ku_j(t)$
  converges.
\end{lemma}
\begin{proof}
The proof easily follows from the classical comparison test for series with terms of constant sign. In fact,
\begin{align*}
&0 \leq \sum_{t=0}^\infty(1-\lambda)^tu_j(t) = \sum_{t=0}^{\infty}(1-\lambda)^t \frac{y_j(t)}{\lVert {\bf y}(t)\lVert_1}\\
&\leq \sum_{t=0}^{\infty}(1-\lambda)^t \frac{1}{\lVert {\bf y}(0)\lVert_1} =  \frac{1}{\lVert {\bf y}(0)\lVert_1} \sum_{t=0}^{\infty}(1-\lambda)^t =  \frac{1}{\lambda \lVert {\bf y}(0)\lVert_1}
\end{align*}
where we have exploited the fact that $y_j(t) \in (0,1]$ and that $\lVert {\bf y}(t) \lVert_1 \geq \lVert {\bf y}(0) \lVert_1, \, \forall t \geq 0,$ as consequence of Lemma \ref{monotone}.
$\square$
\end{proof}

\begin{lemma} \label{MYirr} 
Let  $M\in [0,1]^{N\times N}$ be an irreducible row stochastic matrix, and assume that ${\bf y}\in (0,1]^N$.
Then
\be
M {\bf y} \le {\bf y} \qquad
\Rightarrow \qquad
\left\{
\begin{array}{l}
M {\bf y} = {\bf y} \cr
{\bf y} = \alpha {\bf 1}_N, \ \exists\ \alpha \in (0,1].
\end{array}
\right.
\ee
\end{lemma}

\begin{proof}
If all entries of ${\bf y}$ are identical then 
${\bf y} = \alpha {\bf 1}_N, \ \exists\ \alpha \in (0,1]$, and
it is immediate to see that $M{\bf y} = \alpha M{\bf 1}_N = \alpha {\bf 1}_N = {\bf y}.$
\\
Suppose, now, that not all entries of ${\bf y}$ are identical. It entails no loss of generality assuming that
$y_1= \dots = y_h < y_{h+1} \le \dots \le y_N$, since we can always reduce ourselves to this case by resorting to a suitable permutation.
Then, condition $M {\bf y} \le {\bf y}$ implies
$\sum_{j=1}^h m_{ij} y_j + \sum_{j=h+1}^N m_{ij} y_j \le y_i$
for every $i\in \{1,2,\dots, N\}$. For $i\in \{1,2,\dots, h\}$ the previous inequality becomes
$\sum_{j=1}^h m_{ij} y_1 + \sum_{j=h+1}^N m_{ij} y_j \le y_1,$
that holds true if and only $m_{ij}=0$ for every $i\in \{1,2,\dots, h\}$ and $j\in \{h+1, h+2, \dots, N\}$.
This means that  $M$ is lower block triangular,
 thus contradicting its irreducibility.
Therefore all entries of ${\bf y}$ are equal and ${\bf y}$ is an eigenvector of $M$ corresponding to the unitary eigenvalue.
$\square$ \end{proof}

\bibliographystyle{plain} 

\bibliography{Refer167}

\begin{thebibliography}{10}

\bibitem{Altafini2013}
C.~Altafini.
\newblock Consensus problems on networks with antagonistic interactions.
\newblock {\em IEEE Trans. Aut. Contr.}, 58 (4):935--946, 2013.

\bibitem{BulloML}
O.~Askarisichani, E.~Y. Huang, K.~K. Sato, N.~E. Friedkin, F.~Bullo, and A.~K.
  Singh.
\newblock Expertise and confidence explain how social influence evolves along
  intellective tasks.
\newblock {\em arXiv 2011.07168v1}, pages 1--16.

\bibitem{DeGroot}
M.~H. DeGroot.
\newblock Reaching a consensus.
\newblock {\em Journal of the American Statistical Association},
  69(345):118--121, 1974.

\bibitem{sociological_model}
E.~Gorbatikov, E.~Kornilina, A.~Mikhailov, and A.~Petrov.
\newblock Mathematical model of opinion dynamics in social groups.
\newblock {\em Mediterranean Journal of Social Sciences}, 4(10):380--387, 2013.

\bibitem{Heider}
F.~Heider.
\newblock Social perception and phenomenal causality.
\newblock {\em Psycological Review}, 51 (6):358--374, 1944.

\bibitem{assign_and_appraise}
E.Y. Huang, D.~Paccagnan, W.~Mei, and F.~Bullo.
\newblock Assign and appraise: achieving optimal performance in collaborative
  teams.
\newblock {\em IEEE Trans. Automatic Control}, 63(9):2898--2912, 2018.

\bibitem{collectivelearning_TAC}
W.~Mei, N.E. Friedkin, K.~Lewis, and F.~Bullo.
\newblock Dynamic models of appraisal networks explaining collective learning.
\newblock {\em arXiv:2008.09817v1}.

\bibitem{Mohar}
B.~Mohar.
\newblock The {L}aplacian spectrum of graphs.
\newblock {\em Graph Theory, Combinatorics, and Applications}, 2:871--898,
  1991.

\bibitem{sociology1}
J.~Scott.
\newblock Social network analysis.
\newblock {\em Sociology}, 22(1):397--411, 1988.

\bibitem{sethi}
S.~P. Sethi and G.~L. Thompson.
\newblock {\em Optimal Control Theory: Applications to Management Science and
  Economics}.
\newblock Kluwer Academic Publishers, 2000.

\bibitem{TMS}
D.~M. Wegnar.
\newblock A contemporary analysis of the group mind.
\newblock {\em Theories of Group Behaviour}, pages 185--208, 1987.

\end{thebibliography}

\end{document}